\documentclass[journal]{IEEEtran}

\usepackage{graphicx}
\usepackage{citesort}
\usepackage{amssymb}
\usepackage{amsfonts}
\usepackage{amsmath}
\usepackage{epsfig}
\usepackage{color}
\usepackage{fancybox}
\usepackage{textcomp}
\usepackage{multirow}
\usepackage{setspa ce}
\usepackage{psfrag}
\usepackage[ruled,vlined,linesnumbered]{algorithm2e}

\usepackage{mathrsfs}

\usepackage{booktabs}
\usepackage{threeparttable}
\newtheorem{Theorem}{Theorem}

\newcommand{\argmax}{\operatorname*{argmax}}

    \def\Complex{{\rm\rule[.23ex]{.03em}{1.1ex}\kern-.3em{C}}}

    \newcommand{\be}{\begin{equation}} \newcommand{\ee}{\end{equation}}
    \newcommand{\bea}{\begin{eqnarray}} \newcommand{\eea}{\end{eqnarray}}
    \newcommand{\benum}{\begin{enumerate}} \newcommand{\eenum}{\end{enumerate}}

    \newcommand{\qc}{{\bf c}}

    \newcommand{\qg}{{\bf g}}

    \newcommand{\qm}{{\bf m}}

    \newcommand{\qs}{{\bf s}}

    \newcommand{\qy}{{\bf y}}

    \newcommand{\qY}{{\bf Y}}

    \newcommand{\qepsilon}{{\boldsymbol \epsilon}}

    \newcommand{\bbR}{{\mathbb R}}
    \newcommand{\bbC}{{\mathbb C}}

\begin{document}

\title{Active RIS-Assisted MIMO System for Vital Signs Extraction: ISAC Modeling, Deep Learning, and Prototype Measurements}

\author{De-Ming Chian, \IEEEmembership{Member, IEEE},~Chao-Kai~Wen, \IEEEmembership{Fellow, IEEE},~Feng-Ji~Chen,
Yi-Jie~Sun, and~Fu-Kang~Wang, \IEEEmembership{Member, IEEE}
\thanks{D.-M.~Chian, C.-K.~Wen, and F.-J.~Chen are with the Institute of Communications Engineering, National Sun Yat-sen University, Kaohsiung 804, Taiwan (e-mail: $\rm icefreeman123@gmail.com$, $\rm chaokai.wen@mail.nsysu.edu.tw$, $\rm king19635@gmail.com$).}
\thanks{Y.-J.~Sun and F.-K.~Wang are with the Department of Electrical Engineering, National Sun Yat-sen University, Kaohsiung 804, Taiwan (e-mail: $\rm spring1968bear@gmail.com$, $\rm fkw@mail.ee.nsysu.edu.tw$).}
\thanks{ \color{red}{The source code for DMTNet will be released soon.} }
}

%\markboth{IEEE Transactions on Wireless Communications,~Vol.~XX, No.~XX, XXX~2013}%
%{Shell \MakeLowercase{\textit{et al.}}: Bare Demo of IEEEtran.cls for Journals}

% make the title area
\maketitle

\begin{abstract}
We present the RIS-VSign system, an active reconfigurable intelligent surface (RIS)-assisted multiple-input multiple-output orthogonal frequency division multiplexing (MIMO-OFDM) framework for vital signs extraction under an integrated sensing and communication (ISAC) model. The system consists of two stages: the phase selector of RIS and the extraction of vital signs. To mitigate synchronization-induced common phase drifts, the difference of Möbius transformation (DMT) is integrated into the deep learning framework, named DMTNet, to jointly configure multiple active RIS elements. Notably, the training data are generated in simulation without collecting real-world measurements, and the resulting phase selector is validated experimentally. For sensing, multi-antenna measurements are fused by the DC-offset calibration and the DeepMining-MMV processing with CA-CFAR detection and Newton's refinements. Prototype experiments indicate that active RIS deployment improves respiration detectability while simultaneously enabling higher-order modulation; without RIS, respiration detection is unreliable and only lower-order modulation is supported.
\end{abstract}


\section{Introduction}

Driven by the vision of sixth-generation (6G) wireless networks, ISAC is widely regarded as a cornerstone technology for enabling native environmental perception and context awareness on top of high-throughput connectivity. By unifying sensing and communications, ISAC allows the two functions to share spectrum, waveforms, and infrastructure, thereby meeting requirements for spectral and hardware efficiency, low latency, high reliability, and deployability in realistic environments with blockage and multipath. However, in practical MIMO-OFDM systems, non-contact vital-sign sensing is often fragile due to weak micro-motion returns and synchronization impairments \cite{Zeng-2019,He-2020}, e.g., sampling frequency offset (SFO), carrier frequency offset (CFO), and packet detection delay (PDD), which induce common phase drifts.

To be effective in real-world ISAC systems, an RIS with the improvement ability of signal quality is suitable to support fastly because of the low control overhead.  Despite rapid progress, practical ISAC deployments face a persistent control bottleneck, particularly when a RIS is introduced as a programmable propagation interface  \cite{Chepuri-2024}. 
Optimization-based phase configuration and pilot-intensive strategies can be computationally expensive and sensitive to phase drift, while purely data-driven solutions are often limited by the scarcity of labeled real-world measurements and potential domain mismatch between simulation and deployment. These challenges motivate a model-informed learning approach that (i) reduces reliance on extensive calibration and data collection, and (ii) maintains sensing robustness without sacrificing communication performance.

\begin{figure}
    \centering
    \resizebox{3.2in}{!}{%
    \includegraphics*{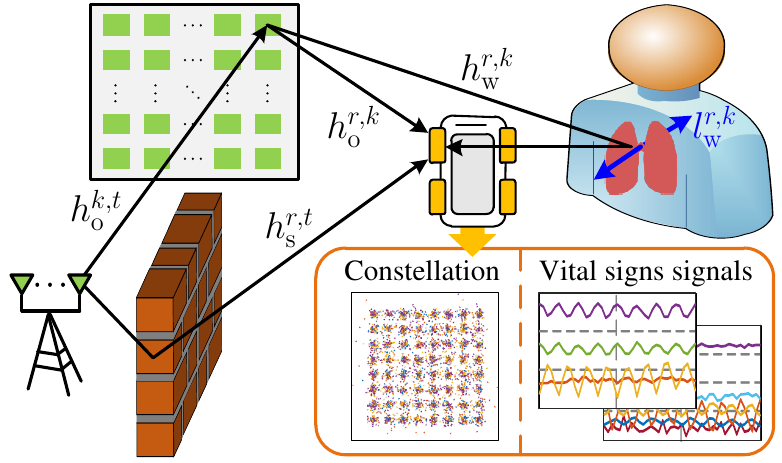} }%
    \caption{Application of RIS-VSign system.
    \label{fig:RIS-VSign}}
\end{figure}

The proposed RIS-VSign system consists of two steps: the phase selector of RIS and vital signs extraction. The key contributions of this study are as follows:
\begin{itemize}
\item The proposed phase selector of RIS is based on the deep learning, and verified by the experiments. 
\item Without using any real-world data, the training dataset is constructed based on a modified channel model \cite{Chian-2023} in the simulation. This implies that the proposed ISAC model is suitable for practical deployment.
\item We evaluate that the deployment of an active RIS in a MIMO system simultaneously enhances the experimental performance of both sensing and communication.
\end{itemize}

% Comment: 這節我幾乎看不懂
% 我認為 應該畫個圖把 (3) 的路徑哪一段到哪裡講清楚
% 模型是要清楚，不是把所有東西都放上去，例如 the phase offset 可能不應該放上去，就說可以去除就可以，或是之後再加上去
% Ans: 學生修改成加在後面。
% 我認為本節應該從 We consider a 5G NR-compatible OFDM ... 開始
% 前面三段，我根本不知道要說甚麼
% Ans: 是講模型建立的概念，學生刪掉了，如果後續introduction有需要的話，再描述。
\section{ISAC Model}

\subsection{Practical Channel Model}

In Fig.~\ref{fig:RIS-VSign}, we consider a 5G NR-compatible OFDM downlink system, where a Tx equipped with $T$ transmit antennas and $K$ active RIS elements serve an Rx with $R$ receive antennas.
We assume that the echo signal from a person is only influenced by the RIS AM path, implying that the static channels do not contain vital signs.
For respiration rate estimation, $P$ consecutive packets are required \cite{He-2020}.

The static components of the channel include the RIS SM path and the direct Tx-Rx path, while the dynamic components include the RIS AM path and the Doppler-affected reflections.
The channel frequency response (CFR) of the static channel, observed at the $r$-th receive antenna from the $t$-th transmit antenna, is given by $h^{r, t}_{\rm s}$.
Moreover, the CFR of the dynamic channel from the $k$-th active RIS element is given by
\begin{equation} \label{eq:hdk}
    h^{r, t}_{\rm d, \theta_k}
    =
    (
    h^{r, k}_{\rm o}
    +
    h^{r, k}_{\rm w}
    e^{ -j \frac{2 \pi}{\lambda}
    l^{r, k}_{\rm w}  } )  a_k h^{k, t}_{\rm o},
\end{equation}
where $a_k = g_k e^{ -j \theta_k }$ is the response of the $k$-th active RIS element with a constant gain $g_k$ and a phase shift $\theta_k$, $h^{r, k}_{\rm o}$ and $h^{r, k}_{\rm w}$ represent the complex attenuation components without and with vital signs, respectively, $l^{r, k}_{\rm w}$ denotes the chest wall displacement, and $\lambda$ is the wavelength.

However, when applying the channel models in practical scenarios, an additional challenge arises: the lack of perfect time synchronization between Tx and Rx. This results in time-varying phase offsets \cite{He-2020}, including SFO, CFO, and PDD.
%sampling frequency offset (SFO), carrier frequency offset (CFO), and packet detection delay (PDD).
Since all Rx antennas share the same RF oscillator, they experience the same phase offsets when receiving from a given Tx antenna, making the phase offset independent of the Rx antennas.
Then, the CFRs of the static and dynamic channels considering time-varying phase offsets are given by
$c^{r, t}_{\rm s} = h^{r, t}_{\rm s}e^{ -j \psi^t }$ and $c^{r, t}_{\rm d, \theta_k} = h^{r, t}_{\rm d, \theta_k} e^{ -j \psi^t }$, respectively,
where $\psi^{t}$ represents the time-varying phase offset.
Assuming an unit transmission power per antenna, the signal power of the $r$-th receivr antenna with the $t$-th transmit antenna and the $k$-th active RIS element is defined as:
\begin{equation} \label{eq:RP}
  E^{r, t}_{\theta_k}
  = | {c^{r, t}_{\rm s}} + {c^{r, t}_{\rm d, \theta_k}} |^2
  = | {h^{r, t}_{\rm s}} + {h^{r, t}_{\rm d, \theta_k}} |^2.
\end{equation}

\subsection{Optimization Model}

To simplify the optimization problem of $E$, we consider a single RIS element. For the $k$-th active RIS element with phase shift $\theta_k$, the maximization of (\ref{eq:RP}) in a MIMO system is expressed as:
\begin{equation} \label{eq:SNRone}
  \sum_{r=1}^{R}
  \sum_{t=1}^{T}
  E^{r, t}_{\theta_k}
  \triangleq
  \sum_{r=1}^{R}
  \sum_{t=1}^{T}
  | { h^{r, t}_{\rm s} } |^2 +
  | { h^{r, t}_{\rm d, \theta_k} } |^2 +
  2 \mathfrak{Re}(p^{r, t}_{\theta_k}),
\end{equation}
where $p^{r, t}_{\theta_k} = { h^{r, t}_{\rm s} }  ( { h^{r, t}_{\rm d, \theta_k} } )^* $ is the complex inner product of  ${ h^{r, t}_{\rm s} }$ and $h^{r, t}_{\rm d, \theta_k}$.
The first term in (\ref{eq:SNRone}) is a constant. Moreover, the chest wall displacement during respiration is minimal, and the received signal magnitude from the human body is significantly lower than that of the direct path, i.e., $|  h^{r, k}_{\rm o} | > | h^{r, k}_{\rm w} |$ in (\ref{eq:hdk}). Therefore, during the process of the maximization of (\ref{eq:RP}), we reasonably assume that the second term in (\ref{eq:SNRone}) is also constant. Then, the maximization of (\ref{eq:SNRone}) is equivalent to maximize the reward function:
\begin{equation} \label{eq:reward}
  Q({\theta_k})
  =
  \sum_{r=1}^{R}
  \sum_{t=1}^{T}
  \mathfrak{Re}
    \left(
    p^{r, t}_{\theta_k}
    \right).
\end{equation}

Furthermore, the complex inner product $p^{r, t}_{\theta_k}$ is rewritten as
\begin{equation} \label{eq:InproSensing}
  p^{r, t}_{\theta_k}
    =
   e^{ j \theta_k }
   \left[
    (h^{r, t}_{\rm s})^*
    h^{k, t}_{\rm o}
    h^{r, k}_{\rm o}
    +
    (h^{r, t}_{\rm s})^*
    h^{k, t}_{\rm o}
    h^{r, k}_{\rm w}
    e^{ -j \frac{2 \pi}{\lambda}
    l^{r, k}_{\rm w}  }
   \right]^*,
\end{equation}
where the first and second terms in the square bracket correspond to the communication and sensing channels, respectively. As shown in \cite{Chepuri-2024}, RIS can enhance both communication and sensing in an ISAC framework. Since (\ref{eq:InproSensing}) also captures this phenomenon, maximizing (\ref{eq:SNRone}) improves both communication and sensing performance.

% Commnet: 若你最後要得到 (9) ，其實根本不需要寫 MIMO 引入 R、T、K
% 最後我不太知道你的目的要做甚麼?
% Ans: R,T是為了優化問題的加總，和後面說明生理訊號的數量

\section{Phase Selector}

To eliminate time-varying offsets, two strategies are summarized in \cite{He-2020}. One utilizes an accurate phase information by using the direct coaxial connection between Tx and Rx. The other is based on the Möbius transformation (MT) \cite{Zeng-2019}. Because of the random location of Rx, the direct coaxial connection is not suitable for practical environments. We let the $r'$-th receive antenna be the reference. Using the $k$-th RIS element with the $\theta_k$ phase shift, the MT of total channel is applied and defined as follows:
\begin{subequations} \label{eq:MT}
\begin{align}
    m^{r, t}_{\theta_k}
    & \triangleq
    ( c^{r', t}_{\rm s} + c^{r', t}_{{\rm d}, \theta_k} ) /
    ( c^{r, t}_{\rm s} + c^{r, t}_{{\rm d}, \theta_k} ) \\
    & =
    (\mathcal{A} e^{ j \theta_k} + \mathcal{B}) /
    (\mathcal{C} e^{ j \theta_k} + \mathcal{D}) \\
    & =
    \mathcal{M}[\mathcal{C} e^{ j \theta_k} + \mathcal{D})]^{-1} +
    ( \mathcal{A} / \mathcal{C} ),
\end{align}
\end{subequations}
where
\begin{equation} \label{eq:MTpar}
\begin{split}
    \mathcal{A}
    & =
    h^{r', t}_{\rm s}
    , \;
    \mathcal{B} =
    g_k
    h^{k, t}_{\rm o}
    ( h^{r', k}_{\rm o} +
    h^{r', k}_{\rm w}
    e^{ -j \frac{2 \pi}{\lambda}
    l^{r', k}_{\rm w}  } ),  \\
    \mathcal{C}
    & =
    h^{r, t}_{\rm s}
    , \;
    \mathcal{D} =
    g_k
    h^{k, t}_{\rm o}
    ( h^{r, k}_{\rm o} +
    h^{r, k}_{\rm w}
    e^{ -j \frac{2 \pi}{\lambda}
    l^{r, k}_{\rm w}  } ),  \\
    \mathcal{M}
    & =
    (\mathcal{B}\mathcal{C} - \mathcal{A}\mathcal{D}) / \mathcal{C}
    , \;
    \mathrm{and} \;
    \mathcal{B}\mathcal{C} - \mathcal{A}\mathcal{D} \neq 0.
\end{split}
\end{equation}
The MT of static channel is $m^{r, t}_{\rm s} = c^{r', t}_{\rm s} / c^{r, t}_{\rm s} = \mathcal{A} / \mathcal{C}$.

\subsection{Predictor for One Element} \label{sec:Predictor}

Next, we propose the selector of one RIS element's phase, based on the prediction of tendency in (\ref{eq:reward}) for different $\theta_k$.
Our predictor needs three inputs, $m^{r, t}_{\rm s}$, $m^{r, t}_{\theta_k}$, and $m^{r, t}_{\theta'_k}$, assuming that $| \theta'_k - \theta_k | = 180^{\circ}$.
Because of the respiration's tiny vibration, we ignore the respiration's effect for $m^{r, t}_{\theta_k}$ and $m^{r, t}_{\theta'_k}$ during the RIS's controlling process.
We define the difference of Möbius transformation (DMT) parameter as follows:
\begin{equation} \label{eq:DMT}
    {\delta}^{r, t}_{\theta_k}
    \triangleq
    m^{r, t}_{\theta_k}
    -
    m^{r, t}_{\rm s}
    =
    \mathcal{M}(\mathcal{C} e^{ j \theta_k} + \mathcal{D})^{-1}.
\end{equation}
To get $p^{r, t}_{\theta_k}$ in (\ref{eq:SNRone}), we use two DMT parameters, ${\delta}^{r, t}_{\theta_k}$ and ${\delta}^{r, t}_{\theta'_k}$, and the following theorem can be found.
\begin{Theorem} \label{thm1}
When $| \theta'_k - \theta_k | = 180^{\circ}$, the complex inner product is
$p^{r, t}_{\theta_k} =
  \mathfrak{Re}
    \left(
    p^{r, t}_{\theta_k}
    \right) +
   j \mathfrak{Im}
    \left(
    p^{r, t}_{\theta_k}
    \right) $,
where
\begin{subequations} \label{eq:InProduct}
\begin{align}
    \mathfrak{Re}
    \left( p^{r, t}_{\theta_k} \right)
    &=
    \frac{
    | \delta^{r, t}_{\theta'_k} |^2 -
    | \delta^{r, t}_{\theta_k} |^2
    }
    { 4 | \delta^{r, t}_{\theta'_k} |^2 }
    E^{r, t}_{\theta_k},
    \label{eq:InProduct_Re}                                                       \\
    j \mathfrak{Im}
    \left( p^{r, t}_{\theta_k} \right)
    &=
    \frac{
    \delta^{r, t}_{\theta'_k}
    (\delta^{r, t}_{\theta_k})^* -
    (\delta^{r, t}_{\theta'_k})^*
    \delta^{r, t}_{\theta_k}
    }
    { 4 | \delta^{r, t}_{\theta'_k} |^2 }
    E^{r, t}_{\theta_k}.
    \label{eq:InProduct_Im}
\end{align}
\end{subequations}
\end{Theorem}
\begin{IEEEproof}[Proof \ref{thm1}]
Please refer to Appendix for more details.
\end{IEEEproof}
With the given $p^{r, t}_{\theta_k}$, we can predict the complex inner product $p^{r, t}_{\theta^{\star}_k}$ of any $\theta^{\star}_k \in [0, 2\pi)$. The corresponding magnitude and phase of $p^{r, t}_{\theta^{\star}_k}$ are respectively expressed as
\begin{subequations} \label{eq:AnyP}
\begin{align}
  | p^{r, t}_{\theta^{\star}_k} |
  & =
  | p^{r, t}_{\theta_k} | , \label{eq:AnyPM} \\
  \measuredangle{p^{r, t}_{\theta^{\star}_k}}
  & =
  \theta^{\star}_k - \theta_k +
  \measuredangle{p^{r, t}_{\theta_k}}. \label{eq:AnyPA}
\end{align}
\end{subequations}
Finally, by substituting (\ref{eq:AnyP}) into (\ref{eq:reward}), $\widehat{ \theta_k } = \argmax\limits_{\theta^{\star}_k} Q({\theta^{\star}_k})$ is the selected phase of the $k$-th RIS element.

Note that the complex inner product $p^{r, t}_{\theta^{\star}_k}$ cannot be directly achieved by the complex inner product between the DMT parameter ${\delta}^{r, t}_{\theta_k}$ and the MT of static channel $m^{r, t}_{\rm s}$ without using Theorem \ref{thm1}. The reason is that ${\delta}^{r, t}_{\theta_k}$ still contains static component from $\mathcal{D}$.

\subsection{DMTNet for Two Elements}

\begin{figure}
    \centering
    \resizebox{3.4in}{!}{%
    \includegraphics*{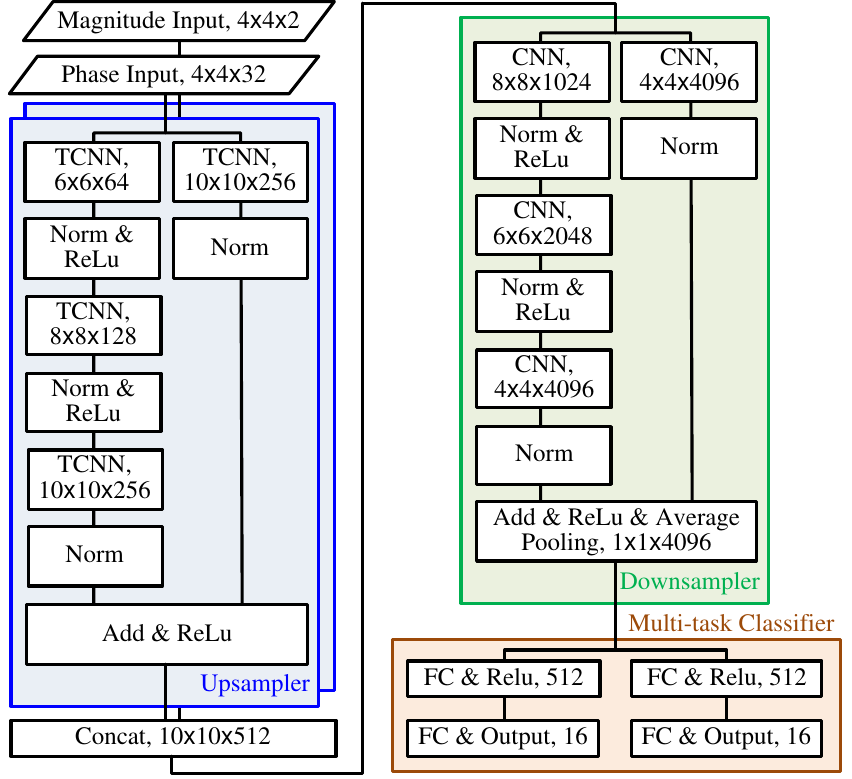} }%
    \caption{DMTNet architecture for 4R4T MIMO system with 2 active RIS elements controlled by 4-bits DPSs.
    \label{fig:DMTNetArchitecture}}
\end{figure}

Because of the complicated interference among RIS elements, the extension of Theorem \ref{thm1} for multiple RIS elements is non-trivial. Fortunately, the predictor in (\ref{eq:AnyP}) describes the correlation between static and dynamic channels for all considering phase shifts.
With DMT parameters of each RIS element, their correlations can be achieved and be utilized by the proposed deep learning, namely DMTNet, to select the phase of all RIS elements simultaneously.

DMTNet architecture and the output size of each layer are shown in Fig.~\ref{fig:DMTNetArchitecture}.
Two paths of inputs are the magnitude and phase of $p^{r, t}_{\theta_k}$, corresponding to the dimensions of $R \times T \times S$. For the phase of input, $S = KD$, where $D$ is the number of controllable phases of DPS.
Because of  (\ref{eq:AnyPM}), the magnitude of input for one of RIS element should be considered one times, i.e., $S = K$.

DMTNet is composed of three components, a upsampler, a downsampler, and a multi-task classifier. The upsampler consists of transposed convolution neural networks (TCNNs) to upsample the much small dimension of input, compared with a common application of image. The downsampler consists of convolution neural networks (CNNs) to extract feature. These two components are deployed by shortcut connections, inspired by the residual network to avoid the vanishing gradient problem. Finally, the multi-task classifier is composed of fully-connect neural networks (FCs). The cross-entropy loss function is utilized to optimize classification model.

\section{Vital Signs Extraction}

\subsection{Vital Signs Signals}

The propose of vital signs extraction is to get the respiration rate $\omega$.
Fixing the phase shifts of all $K$ active RIS elements, (\ref{eq:MT}) at the $p$-th packet is redefined as follows:
\begin{equation} \label{eq:MTvs}
    m^{r, t}_{\omega, p} =
    \underline{\mathcal{M}}
    ( \underline{\mathcal{C}} e^{ j \frac{2 \pi}{\lambda}l^{r', k}_{\rm w}  } +
    \underline{\mathcal{D}} )^{-1} +
    ( \underline{\mathcal{A}} / \underline{\mathcal{C}} ),
\end{equation}
where
\begin{equation} \label{eq:MTvspar}
\begin{split}
    \underline{ \mathcal{A} }
    & =
    h^{r', t}_{\rm s} +
    h^{k, t}_{\rm o}
    h^{r', k}_{\rm o}
    a_k
    , \;
    \underline{\mathcal{B}} =
    h^{k, t}_{\rm o}
    h^{r', k}_{\rm w}
    a_k,  \\
    \underline{\mathcal{C}}
    & =
    h^{r, t}_{\rm s} +
    h^{k, t}_{\rm o}
    h^{r, k}_{\rm o}
    a_k
    , \;
    \underline{\mathcal{D}} =
    h^{k, t}_{\rm o}
    h^{r, k}_{\rm w}
    a_k
    e^{ -j \frac{2 \pi}{\lambda}\Delta  },  \\
    \underline{\mathcal{M}}
    & =
    (\underline{\mathcal{B}} \: \underline{\mathcal{C}} -
     \underline{\mathcal{A}} \: \underline{\mathcal{D}}) /
     \underline{\mathcal{C}}
    , \;
    \underline{\mathcal{B}} \: \underline{\mathcal{C}} -
    \underline{\mathcal{A}} \: \underline{\mathcal{D}} \neq 0,
\end{split}
\end{equation}
$\Delta = l^{r, k}_{\rm w} - l^{r', k}_{\rm w}$ approaches to a constant because the $r$-th and $r'$-th antennas are close, and the displacements $l^{r', k}_{\rm w}$ and $l^{r, k}_{\rm w}$ change with the respiration activity at the $p$-th packet.
\cite{Zeng-2019} verified that $m^{r, t}_{\omega, p}$ with vital signs changes following a circle in complex plane. Then, the respiration is described by the phase of $m^{r, t}_{\omega, p}$. Moreover, except the considering $r$-th receive antenna of $m^{r, t}_{\omega, p}$, the other antennas can be the reference, denoted as the $r'$-th receive antenna. Hence, the phase of (\ref{eq:MTvs}) has $R-1$ combinations. However, when the conjugate of $m^{r', t}_{\omega, p}$ is $m^{r, t}_{\omega, p}$, i.e., $m^{r', t}_{\omega, p} = 1 / m^{r, t}_{\omega, p}$, they are dependent on the same displacement. One of them needs to be considered. For (\ref{eq:MTvs}), the phase vector of all necessary combinations is denoted as $\measuredangle{\qm^{r, t}_{\omega, p}}\in \bbR^{(R-1)/2 }$.

When the tester's respiration rate is $\omega$, the corresponding total channel from the $t$-th transmit antenna to the $r$-th receive antenna at the $p$-th packet is denoted as $c^{r, t}_{\omega, p}$.
The amplitude of $c^{r, t}_{\omega, p}$ is also dependent on the displacement of chest wall, and is denoted as $|c^{r, t}_{\omega, p}|$.
The average of $|c^{r, t}_{\omega, p}|$ and $\measuredangle{\qm^{r, t}_{\omega, p}}$ over $F$ subcarriers are applied and denoted as $|\overline{c^{r, t}_{\omega, p}}|$ and $\measuredangle{\overline{\qm^{r, t}_{\omega, p}}}$, respectively.
For the $t$-th transmit antenna, the average amplitude vector and the average phase vector for $R$ receive antennas are respectively expressed as
\begin{subequations}
\begin{align}
    | \qc^{t}_{\omega, p} |
    &=
    [ |\overline{c^{1, t}_{\omega, p}}|,
    \ldots,
    |\overline{c^{R, t}_{\omega, p}}| ]^T
    \in \bbR^R,
    \label{eq:VSmag}                                                \\
    \measuredangle{\qm^{t}_{\omega, p}}
    & =
    [{\measuredangle{\overline{\qm^{1, t}_{\omega, p}}}}^T,
    \ldots,
    {\measuredangle{\overline{\qm^{R, t}_{\omega, p}}}}^T ]^T
    \in \bbR^{R(R-1)/2 }.
    \label{eq:VSang}
\end{align}
\end{subequations}
For $T$ transmit antennas and $P$ packets, the vital signs matrix is $\qY_{\rm vs} = [\qy_{\rm vs}^1, \ldots, \qy_{\rm vs}^N]^T \in \bbR^{N \times P}$,
where the $n$-th vital signs vector $\qy_{\rm vs}^{n}$ is the stack of (\ref{eq:VSmag}) and (\ref{eq:VSang}) from all transmit antennas, and $N \triangleq TR(R+1)/2$.
Moreover, $\qY_{\rm vs}$ can be regarded as $N$ sensors with $P$ sampling points in the measurement area.

Because the respiration activity has the periodicity, the model of the reflected signal from the human subject can be simplified as a complex exponential function of time.
The $n$-th vital signs vector with the frequency $\omega$ is modeled as
\begin{equation} \label{eq:VSvec}
    \qy^{n}_{\rm vs} = g^{n}_{\rm vs} \qs_{\rm vs}( \omega) + \qepsilon^{n}_{\rm vs},
\end{equation}
where $g^{n}_{\rm vs}$ is the complex coefficient of the vital sign, $\qs_{\rm vs}( \omega) = [1, e^{j \omega}, \ldots, e^{j \omega (P-1)}]^T \in \bbC^{P}$ is the exponential vector, and $\qepsilon^{n}_{\rm vs}\in \bbC^{P}$ is the noise vector with the DC offset.

\subsection{Vital Signs Estimation}

Our signal processing consists of two steps: DC offset calibration and vital signs extraction.
The DC offset calibration is introduced by removing $\qy^{n}_{\rm vs}$ with the $2$-nd degree polynomial curve fitting.
For the vital signs extraction, the design is difficult in a MIMO system due to two reasons.
Firstly, the number of the equivalent sensors $N$ is large, e.g., $N=40$ for the 4R4T MIMO system. Secondly, because the vital sign signals are from the different propagation paths, they are corresponding to the obviously different power levels. Therefore, the challenge is to extract the vital signs signals with the different signal qualities.
According to this two reasons, the data fusion algorithm, namely DeepMining-MMV \cite{Chian-2022}, is applied for the vital signs extraction.

Three stages are utilized in DeepMining-MMV. We modify the first stage to detect vital signs and use the second and third stages to estimate them. Firstly, we merge the normalized frequency spectrums of the reflected signals into one spectrum. To determine whether the subjects exist and are stable in the measurement area or not, the cell-averaging constant false alarm rate (CA-CFAR) \cite{Radar-BOOK} algorithm is utilized for the vital sign detection. If the vital signs is detected successfully, the frequency corresponding to the strongest power in the spectrum is configured to the initial frequency $\widehat{\omega}$.
Secondly, the complex coefficient $\widehat{g^{n}_{\rm vs}}$
%$\widehat{\qg_{\rm vs}}= [\widehat{g^{1}_{\rm vs}}, \ldots, \widehat{g^{N}_{\rm vs}}]^T$
is refined by the least-squares (LS) estimator with the given initial frequency. Thirdly, the initial frequency is refined by using the Newton’s method to the normalized LS estimator based on the power of residual signals $|\sigma^n_{\rm vs}|^2$, where the residual signals are achieved by $\qy^{n}_{\rm vs} := \qy^{n}_{\rm vs} - \widehat{g^{n}_{\rm vs}} \qs_{\rm vs}( \widehat{\omega})$. The third stage is equivalent to minimizing the cost function
\begin{equation} \label{eq:LS-estimator-normal}
    J(\widehat{\qg_{\rm vs}},\widehat{\omega}) =
    \sum_{n=1}^{N}
    \frac{
    -2 {\mathfrak{Re}} ( \widehat{g^n_{\rm vs}} {\qy^n_{\rm vs}}^H \qs(\widehat{\omega}) )
    +  |\widehat{g^n_{\rm vs}}|^2 \left\| \qs( \widehat{\omega} ) \right\|^2  }
   { |\sigma^n_{\rm vs}|^2 },
\end{equation}
% where $|\sigma^n_{\rm vs}|^2$ is the residual signal power of the $n$-th equivalent sensor.
% Furthermore, the accuracy of frequency $\widehat{\omega}$ can be improved by repeating the second and third stage with the cyclic method \cite{Mama-2016} for several times.

\section{Perfomance Evaluation}

\subsection{Training Results}

\begin{figure}
    \centering
    \resizebox{3.4in}{!}{%
    \includegraphics*{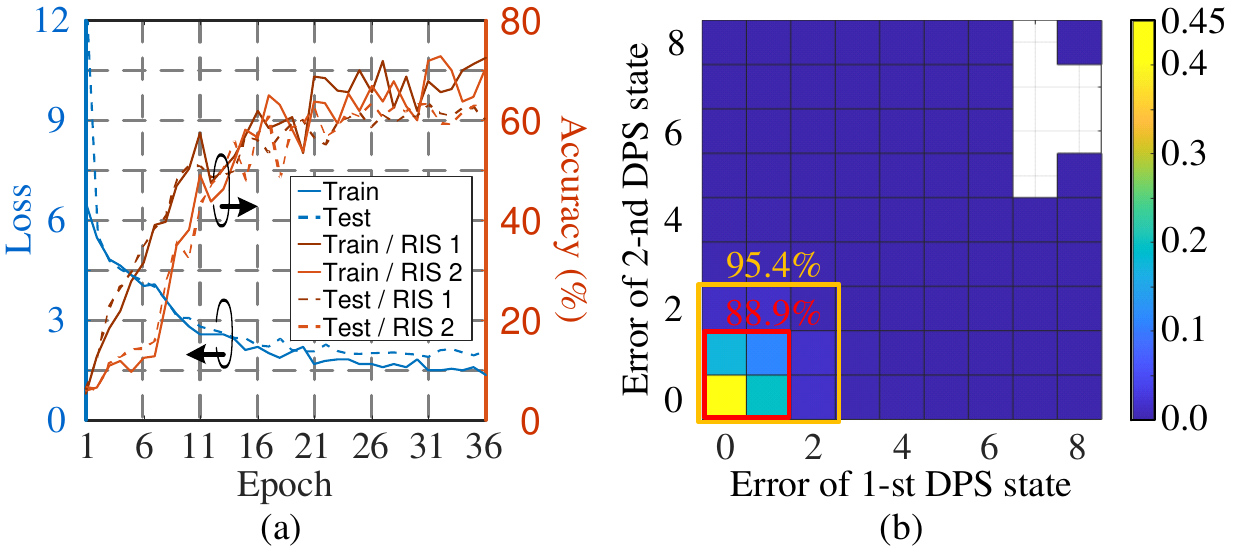} }%
    \caption{DMTNet: (a) Loss and accuracy, and (b) error of testing dataset.
    \label{fig:DMTNetAccuracy}}
\end{figure}

We modify the verified RIS's model \cite{Chian-2023} to a 4R4T MIMO system in the scenario containing scattering clusters based on a large-scale fading. Two RIS elements is applied, and the number of controllable phases $D$ is 16. The training and testing datasets are generated by this channel model without noise. The number of them are 180,000 and 20,000 samples, respectively. The operating frequency of single subcarrier is 3.65 GHz. The optimized loss and accuracy of DMTNet are shown in Fig.~\ref{fig:DMTNetAccuracy}(a), and the errors of two DPSs are shown in Fig.~\ref{fig:DMTNetAccuracy}(b).
Although the accuracy of each DPS is only greater than 60$\%$, approximate 95.4$\%$ of the samples have phase selection errors less than 2. This indicates that, for 95.4$\%$ of the samples, the phase differences from the optimal values are less than 45$^{\circ}$ when considering two RIS elements simultaneously. % 解釋: 因為1個相移器的16個狀態可轉360度，所以選擇相位的最大誤差為180度，因此狀態的最大誤差為8。雖然DMTNet準確度只有60%以上，但是有95.4%的樣本是狀態誤差小於2。這說明使用DMTNet選擇相位的95.4%樣本，跟最佳狀態的相位誤差會小於45度。

\subsection{Measurements}

\begin{figure}
    \centering
    \resizebox{3.5in}{!}{%
    \includegraphics*{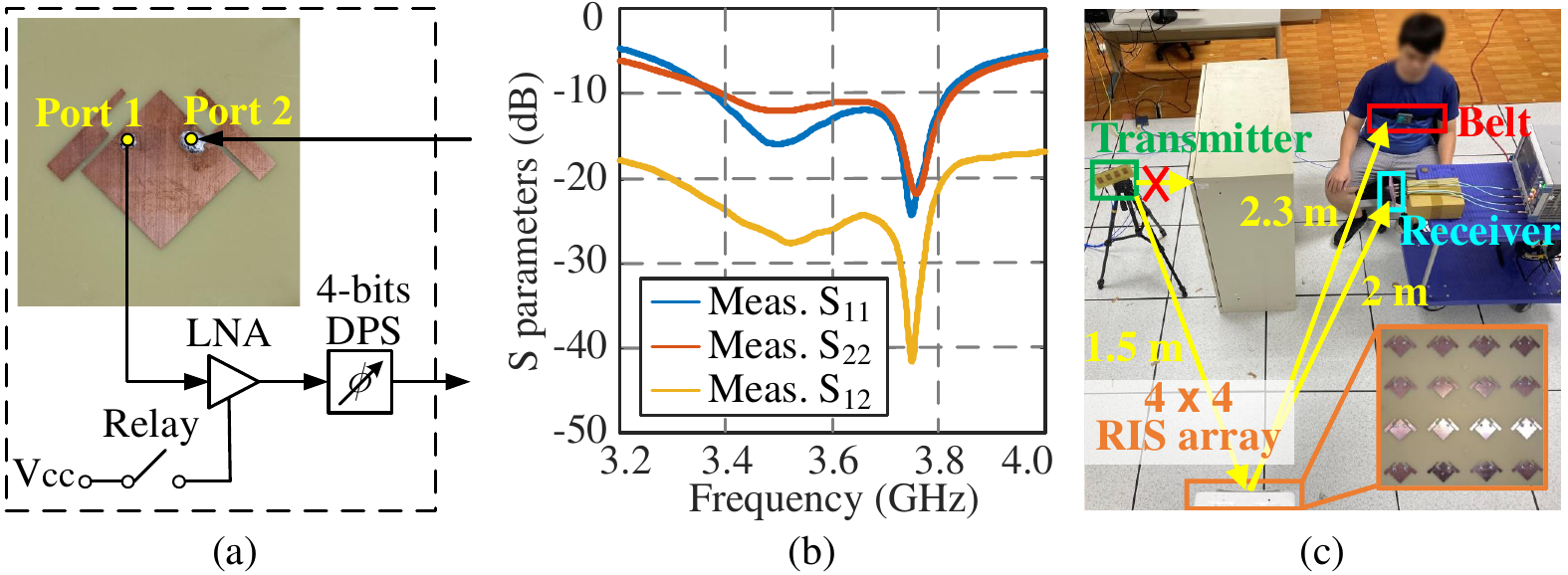} }%
    \caption{(a) Structure of active RIS element, (b) measured S parameters of antenna pair, and (c) experimental scenario.
    \label{fig:RISant}}
\end{figure}

The structure of one active RIS element is shown in Fig.~\ref{fig:RISant}(a), composed of a self-decoupled antenna pair with two feeding ports, a LNA with an approximate 20 dB power gain, a 4-bits DPS with 16 phase shift states in 22.5$^{\circ}$ steps, and a relay.
Two feeding ports perform the function of transmit and receive signals, respectively. They are connected to the different other antenna pairs to enhance the isolation among all ports deployed on a RIS array. Moreover, a self-decoupled antenna pair is operated in the frequency range of 3.4-3.8 GHz. In Fig.~\ref{fig:RISant}(b), their measured reflection and transmission coefficients in the operating band are less than -10 dB and -20 dB, respectively.
Finally, the 4 $\times$ 4 active RIS array is utilized in the ISAC application, shown in Fig.~\ref{fig:RISant}(c), to improve the signal quality degraded by the obstacle. Our 4R4T MIMO system follows that in \cite{Chian-2024} and transmits 5G OFDM signals in
the 3.7-3.72 GHz band.

\begin{figure}
    \centering
    \resizebox{3.7in}{!}{%
    \includegraphics*{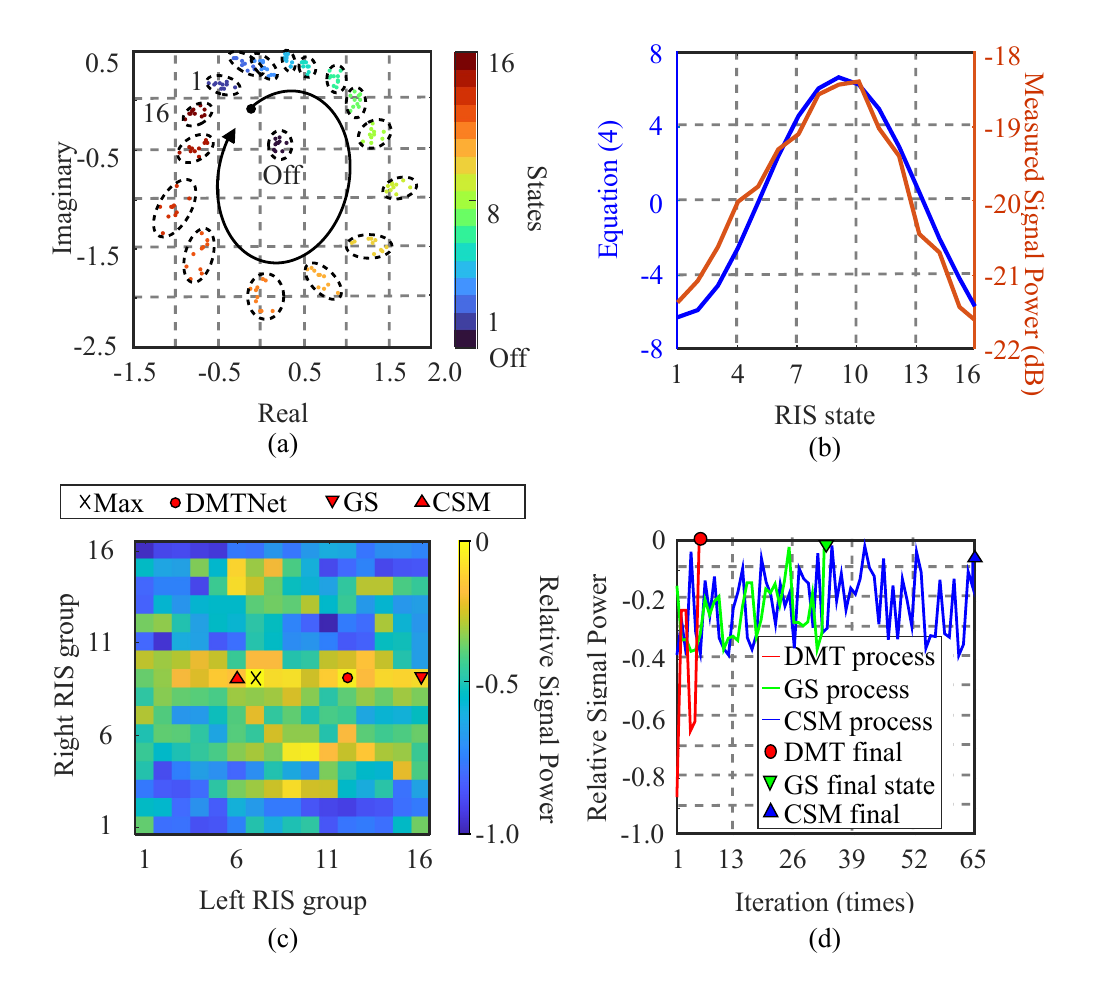} }%
    \caption{(a) $m^{1, 1}_{\rm s}$ and $m^{1, 1}_{\theta_k}$ for all states of DPS.
		     (b) Predicted reward and measuremed signal power.
	     	 (c) All combinations and (d) controlling processes by using different algorithms.
    \label{fig:Meas}}
\end{figure}

Firstly, we verify the predictor in Section \ref{sec:Predictor}. To simplify and strengthen the phenomenon, all RIS elements are specially configured to the same phase shift state, and Rx is in front of the RIS. By turning off all relays, the MT of static channel $m^{r, t}_{\rm s}$ is achieved. When we turn on all relays and iterate through the phase shift states, the MT of total channels for all states $m^{r, t}_{\theta_k}$ are achieved. For the one subcarrier's channel from the 1-st transmit antenna to the 1-st receive antenna, the MT of different states' channels with 10 continuous packets are shown in Fig.~\ref{fig:Meas}(a). The stable and reliable channel information is achieved by averaging continuous packets for each state. According to $| \theta'_k - \theta_k | = 180^{\circ}$, the DMT parameters of the 1-st and 9-th phase shift states are utilized for the predictor, and we will get the reward of each state. By summing across all subcarriers, the rewards are shown in Fig.~\ref{fig:Meas}(b). Because the predicted rewards closely match the measured signal power for each state, the validity of the proposed predictor is confirmed.

Secondly, with the results of predictor, DMTNet is evaluated in Fig.~\ref{fig:Meas}(c). The elements of RIS array are divided into left and right groups. The elements in the same group have the identical phase shift states. For all combinations of states, the relative signal power is shown in Fig.~\ref{fig:Meas}(c). Two other controlling algorithms are included in the discussion. One is the greedy searching (GS) algorithm \cite{Chian-2023}; the other is conditional sample mean (CSM) algorithm \cite{Ren-2022}, based on the phase expectation
of each group. The final states of three algorithms are close to the optimum state. However, their controlling processes are shown in Fig.~\ref{fig:Meas}(d), and the huge difference in their number of iterations can be observed. The number of iterations with DMTNet, GS, and CSM are 5, 32, and 64, respectively.

\begin{figure}
    \centering
    \resizebox{3.1in}{!}{%
    \includegraphics*{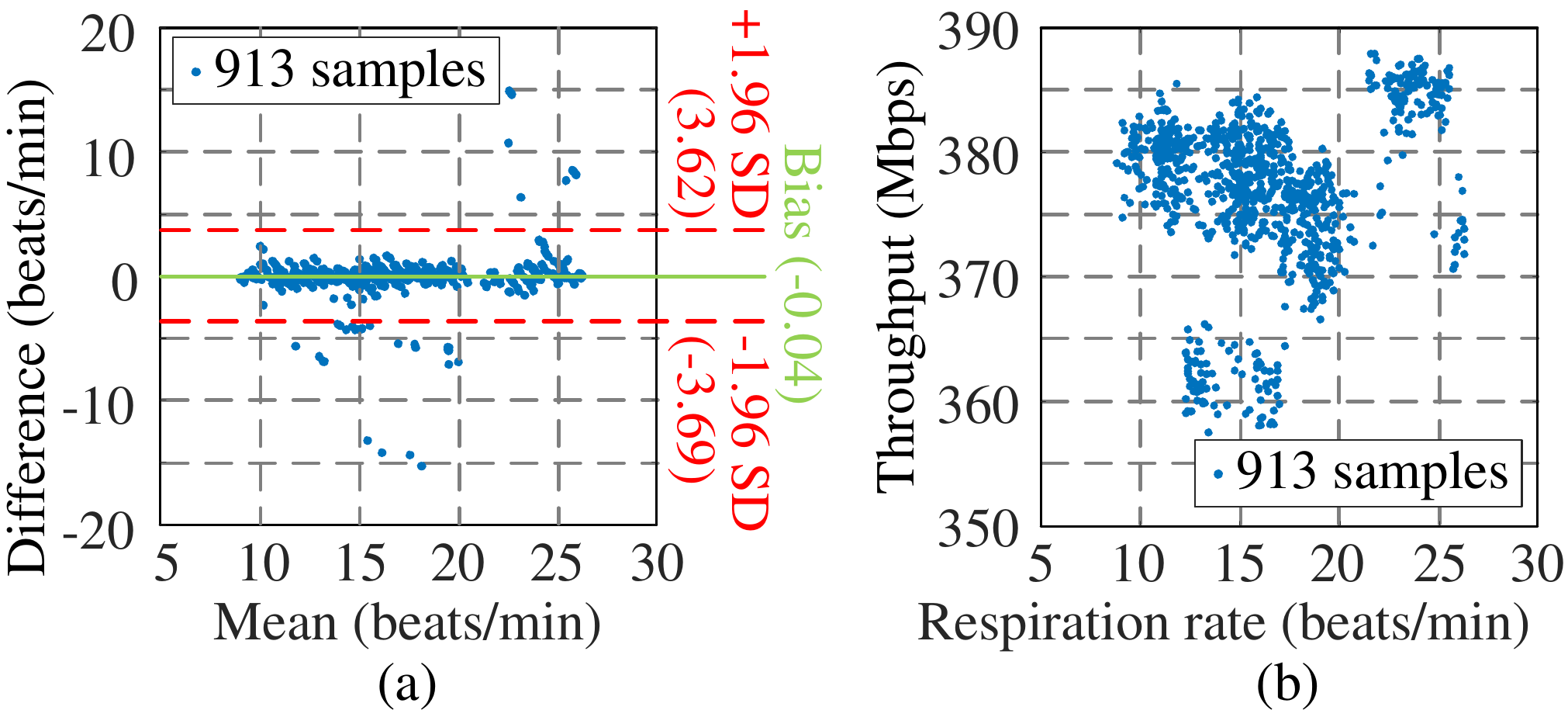} }%
    \caption{ISAC performance:
		     (a) Bland–Altman plot of respiration rate and (b) throughput in different respiration rates.
    \label{fig:VitalSign}}
\end{figure}

After applying the phase selector of DMTNet, the vital signs extraction is utilized. 
We use a respiration belt, as a contact method, to compare with our noncontact method. Because a high correlation does not imply a good agreement, we use the Bland–Altman plot to describe an agreement and shown in Fig.~\ref{fig:VitalSign}(a). The samples of over 95.7$\%$ fall within 1.96 standard deviations from the mean. Because of a good agreement between two methods, a high accuracy of our method is indicated.
Moreover, our MIMO system can support 64 QAM and achieve the maximum throughput of about 388 Mbps in Fig.~\ref{fig:VitalSign}(b).
Contrarily, without RIS, the system cannot detect respiration rate, and the MIMO system only supports 16 QAM.

\section{Conclusion}
The proposed RIS-VSign system demonstrates an ISAC framework in a 5G NR–compatible MIMO-OFDM system by combining DMTNet-based active-RIS phase control with DeepMining-MMV-based respiration estimation. The experimental results indicate that deploying an active RIS can simultaneously enhance sensing reliability and communication performance. In contrast, the RIS-free system fails to reliably detect respiration and supports only lower-order modulation.

\section*{Appendix: Proof of Theorem \ref{thm1}}

With (\ref{eq:SNRone}) and (\ref{eq:DMT}), we have $|{\delta}^{r, t}_{\theta_k}|^{-2} = E^{r, t}_{\theta_k} |\mathcal{M}|^{-2}$.
Then, $|{\delta}^{r, t}_{\theta_k}|^{-2} - |{\delta}^{r, t}_{\theta'_k}|^{-2} = 2[{\mathfrak{Re}}(p^{r, t}_{\theta_k}) - {\mathfrak{Re}}(p^{r, t}_{\theta'_k})] |\mathcal{M}|^{-2}$.
Assuming that $|\theta'_k - \theta_k| = 180^{\circ}$, we can get ${\mathfrak{Re}}(p^{r, t}_{\theta'_k}) = - {\mathfrak{Re}}(p^{r, t}_{\theta_k})$.
Thus, $4{\mathfrak{Re}}(p^{r, t}_{\theta_k}) = |\mathcal{M}|^{2} (|{\delta}^{r, t}_{\theta_k}|^{-2} - |{\delta}^{r, t}_{\theta'_k}|^{-2})  = |\mathcal{M}|^{2} (|{\delta}^{r, t}_{\theta'_k}|^{2} - |{\delta}^{r, t}_{\theta_k}|^{2}) / ( |{\delta}^{r, t}_{\theta_k}|^{2} |{\delta}^{r, t}_{\theta'_k}|^{2})$, where $|\mathcal{M}|^{2} = |{\delta}^{r, t}_{\theta_k}|^{2} E^{r, t}_{\theta_k}$.
Finally, the (\ref{eq:InProduct_Re}) of Theorem \ref{thm1} is proved.

(\ref{eq:DMT}) can be expressed as $({\delta}^{r, t}_{\theta_k})^{-1}\mathcal{M} = \mathcal{C} e^{ j \theta_k} + \mathcal{D}$.
When $|\theta'_k - \theta_k| = 180^{\circ}$, we get $(\mathcal{C} e^{ j \theta_k} + \mathcal{D})(\mathcal{C} e^{ j \theta'_k} + \mathcal{D})^{*} = - |\mathcal{C}|^{2} + |\mathcal{D}|^{2} + 2j{\mathfrak{Im}}(p^{r, t}_{\theta_k}) = |\mathcal{M}|^{2} [{\delta}^{r, t}_{\theta_k}({\delta}^{r, t}_{\theta'_k})^{*}]^{-1}$.
Then, $4j{\mathfrak{Im}}(p^{r, t}_{\theta_k}) = |\mathcal{M}|^{2} \{ [{\delta}^{r, t}_{\theta_k}({\delta}^{r, t}_{\theta'_k})^{*}]^{-1} - [{\delta}^{r, t}_{\theta'_k}({\delta}^{r, t}_{\theta_k})^{*}]^{-1} \} = |\mathcal{M}|^{2} [ \delta^{r, t}_{\theta'_k} (\delta^{r, t}_{\theta_k})^* - (\delta^{r, t}_{\theta'_k})^* \delta^{r, t}_{\theta_k}] / ( |{\delta}^{r, t}_{\theta_k}|^{2} |{\delta}^{r, t}_{\theta'_k}|^{2})$, where $|\mathcal{M}|^{2} = |{\delta}^{r, t}_{\theta_k}|^{2} E^{r, t}_{\theta_k}$.
Finally, the (\ref{eq:InProduct_Im}) of Theorem \ref{thm1} is proved.

\bibliographystyle{IEEEtran}
\bibliography{References_Quantized}

\end{document}